%% file: PerfectQPIR.tex
\pdfoutput=1
\def\draft{0}
\def\llncs{0}
\def\anon{0}
\def\forcepage{0}

\ifnum\llncs=0
\documentclass[11pt]{article}
\usepackage{fullpage}
\else
\documentclass[orivec,runningheads,envcountsame,envcountsect]{llncs}
\usepackage{amssymb,amsmath,cite,url}
\fi

\usepackage{float}
\floatstyle{boxed}
\restylefloat{figure}

\ifnum\draft=1
\def\ShowAuthNotes{1}
\else
\def\ShowAuthNotes{0}
\fi

\usepackage{bm,bbm}

\input{style}

\usepackage{qcircuit}
\usepackage{mathrsfs}

\title{On Quantum Advantage in Information Theoretic Single-Server PIR}

\author{}

\ifnum\llncs=1
\institute{}
\ifnum\anon=0
\author{Dorit Aharonov\inst{1} \and Zvika Brakerski\inst{2} \and Kai-Min Chung\inst{3} \and Ayal Green\inst{1} \and Ching-Yi Lai\inst{5} \and Or Sattath\inst{4}}
\institute{Hebrew University \and Weizmann Institute of Science \and Academia Sinica \and Ben-Gurion University}
\authorrunning{Aharonov, Brakerski, Chung, Green, Lai and Sattath}
\fi
\date{}
\else
\ifnum\anon=0
\author{Dorit Aharonov\thanks{Hebrew University of Jerusalem.} \and Zvika Brakerski\thanks{Weizmann Institute of Science.
} \and Kai-Min Chung\thanks{Academia Sinica.} \and Ayal Green\footnotemark[1] \and Ching-Yi Lai\thanks{National Chiao Tung University.}  \and Or Sattath\thanks{Ben-Gurion University.}}
\fi
\ifnum\draft=1
\date{\today}
\else
\date{}
\fi

\fi

\ifnum\llncs=1
\renewcommand{\paragraph}{\boldpar}
\let\doendproof\endproof
\renewcommand\endproof{\hfill\qed\doendproof}
\fi

\begin{document}

\maketitle

\input{abstract}

\ifnum\forcepage=1
\thispagestyle{empty}
\newpage
\pagenumbering{arabic}
\fi

\input{intro}

\input{prelim}

\input{model}

\input{lowerbound}

\input{acknowledgments}

\bibliographystyle{alphaabbrurldoieprint}
\bibliography{QPIR}

\appendix
\input{stdprelim}

\input{logscheme}

%
%
%

\end{document}

%% file: style.tex
\newcommand{\remove}[1]{}
\newcommand{\ignore}[1]{}
\ifnum\llncs=0
\usepackage{amsthm}
\else
\fi

\usepackage{amsmath,amssymb}
\usepackage{url}
\usepackage{cite}

\usepackage{color}
\definecolor{DarkBlue}{RGB}{0,0,150}
\definecolor{DarkGreen}{RGB}{0,150,0}
\definecolor{Brown}{RGB}{150,75,0}
\definecolor{Yellowish}{RGB}{180,180,0}

\usepackage[colorlinks,linkcolor=DarkBlue,citecolor=DarkBlue]{hyperref}

\ifnum\llncs=0

\theoremstyle{definition}

\newtheorem{theorem}{Theorem}[section]

\newtheorem{definition}[theorem]{Definition}

\newtheorem{lemma}[theorem]{Lemma}

\newtheorem{remark}{Remark}
\newtheorem{corollary}[theorem]{Corollary}
\else

\spnewtheorem{prop}{Property}{\bfseries}{\itshape}
\spnewtheorem{fact}{Fact}{\bfseries}{\itshape}
\spnewtheorem{subclaim}{Claim}[theorem]{\bfseries}{\itshape}
\spnewtheorem{tlclaim}[theorem]{Claim}{\bfseries}{\itshape}
\fi

\ifnum\llncs=1
\def\pfend{\hfill\qedsymbol}
\else
\def\pfend{}
\fi

\def\cA{{\cal A}}
\def\cB{{\cal B}}
\def\cC{{\cal C}}

\def\cM{{\cal M}}

\def\cR{{\cal R}}
\def\cS{{\cal S}}

\def\cX{{\cal X}}
\def\cY{{\cal Y}}

\def\bbI{{\mathbb I}}

\def\sA{{\mathscr{A}}}
\def\sB{{\mathscr{B}}}
\def\sF{{\mathscr{F}}}
\def\sI{{\mathscr{I}}}

\def\tr{\textnormal{tr}}

\def\binset{\{0,1\}}

\def\q2{\lfloor q/2 \rceil}

\newcommand{\abs}[1]{\left\vert {#1} \right\vert}

\newcommand{\mx}[1]{\mathbf{{#1}}}
\newcommand{\vc}[1]{\mathbf{{#1}}}

\newcommand{\zo}{\{0,1\}}

\newcommand{\boldpar}[1]{\vspace{3pt}\par\noindent\textbf{#1}}

\ifnum\ShowAuthNotes=1
\newcommand{\authnote}[3]{\textcolor{#3}{[{\footnotesize {\bf #1:} { {#2}}}]}}
\else
\newcommand{\authnote}[3]{}
\fi

\newcommand{\onote}[1]{\authnote{O}{#1}{cyan}}

\newcommand{\knote}[1]{\authnote{K}{#1}{magenta}}

\newcommand{\db}{{\mathtt{DB}}}

\newcommand{\absnewline}{\ifnum\llncs=1 \\ \fi}

\def\mgp[#1]{\mx{G}_{#1}}

\def\mgip[#1]{\mx{G}_{#1}^{-1}}

\def\mcit[#1]{\mci[#1]^T}
\def\mci[#1]{\mx{C}_{#1}}

\newcounter{hybridcount}
\newcounter{prevhybridcount}
\newcounter{nexthybridcount}

\def\beginM{\left[\begin{matrix}}
\def\endM{\end{matrix}\right]}

\def\m0{\mx{0}}

\newcommand{\ket}[1]{|{#1}\rangle}
\newcommand{\bra}[1]{\langle{#1}|}

\newcommand{\dblen}{n}
\newcommand{\dbllen}{\ell}

\newcommand{\auxr}{\vc{r}}

%% file: abstract.tex

\begin{abstract}
	In (single-server) Private Information Retrieval (PIR), a server holds a large database $\db$ of size $n$, and a client holds an index $i \in [n]$ and wishes to retrieve $\db[i]$ without revealing $i$ to the server. It is well known that information theoretic privacy even against an ``honest but curious'' server requires $\Omega(n)$ communication complexity. This is true even if quantum communication is allowed and is due to the ability of such an adversarial server to execute the protocol on a superposition of databases instead of on a specific database (``input purification attack'').
	Nevertheless, there have been some proposals of protocols that achieve sub-linear communication and appear to provide some notion of privacy. Most notably, a protocol due to Le Gall (ToC 2012) with communication complexity $O(\sqrt{n})$, and a protocol by Kerenidis et al. (QIC 2016) with communication complexity $O(\log(n))$, and $O(n)$ shared entanglement.

	We show that, in a sense, input purification is the only potent adversarial strategy, and protocols such as the two protocols above are secure in a restricted variant of the quantum honest but curious (a.k.a specious) model. More explicitly, we propose a restricted privacy notion called \emph{anchored privacy}, where the adversary is forced to execute on a classical database (i.e. the execution is anchored to a classical database). We show that for measurement-free protocols, anchored security against honest adversarial servers implies anchored privacy even against specious adversaries.
	
	Finally, we prove that even with (unlimited) pre-shared entanglement it is impossible to achieve security in the standard specious model with sub-linear communication, thus further substantiating the necessity of our relaxation. This lower bound may be of independent interest (in particular recalling that PIR is a special case of Fully Homomorphic Encryption).
\end{abstract}

%% file: intro.tex

\section{Introduction}
\label{sec:introduction}

\onote{TODO: Add link to the arXiv}

Private Information Retrieval (PIR), introduced by Chor et al.~\cite{CGKS95}, is perhaps the most basic form of joint computation with privacy guarantee. PIR is concerned with privately retrieving an entry from a database, without revealing which entry has been accessed. Formally, a PIR protocol is a communication protocol between two parties, a server holding a large database $\db$ containing ${\dblen}$ binary entries\footnote{Throughout this work we will focus on the setting of binary database. We do note that there is vast literature concerned with optimizations for the case of larger alphabet.}, and a client who wishes to retrieve the $i$th element of the database but without revealing the index $i$. Privacy can be defined using standard cryptographic notions such as indistinguishability or simulation (see \cite{GoldreichFoundationsVolume2}). The simplicity of this primitive is since there is no privacy requirement for the database (i.e.\ we allow sending more information than necessary) and that the server is not required to produce any output in the end of the interaction, so functionality and privacy are \emph{one sided}.

Clearly PIR is achievable by sending all of $\db$ to the client. This will have communication complexity ${\dblen}$ and will be perfectly private under any plausible definition since the client sends no information. The absolute optimal result one could hope for is a protocol with logarithmic communication, matching the most communication efficient protocol without privacy constraints, in which the client sends the index $i$ to the server and receives $\db[i]$ in response.

Alas, \cite{CGKS95} proved that linear (in ${\dblen}$) communication complexity is \emph{necessary} for PIR, and that this is the case even in the presence of arbitrary setup information.\footnote{Setup refers to any information that is provided to the parties prior to the execution of the protocol by a trusted entity, but crucially one that does not depend on the parties' inputs. Shared randomness or shared entanglement are common examples.} Despite its pessimistic outlook, this lower-bound served (already in \cite{CGKS95} itself) as starting point to two extremely prolific and influential lines of research, showing that the communication complexity can be vastly improved if we place some restrictions on the server. The first considered \emph{multiple non-interacting} servers (see, e.g., \cite{E12,DG15} and references therein), instead of just a single server, and the second considered \emph{computationally bounded} servers and relying on \emph{cryptographic assumptions} (see, e.g., \cite{CMS99,G09thesis,BV11lwe}).

While our discussion so far referred to protocols executed by classical parties over classical communication channels, the focus of this work is on the quantum setting, where there is a quantum communication channel between the client and server, and where the parties themselves are capable of performing quantum operations. Importantly, we still only require functionality for a classical database and a classical index.

One could hope that introducing quantum channels could allow an information theoretic solution to a problem that classically can only be solved using cryptographic assumptions, as has been the case for quantum key distribution \cite{BB84}, quantum money \cite{W83}, quantum digital signatures~\cite{GC01}, quantum coin-flipping ~\cite{M07,CK09,ACGKM16} and more~\cite{BS16}.
Indeed, the notion of Quantum PIR (or QPIR) is quite a natural extension of its classical counterpart and has also been extensively studied in the literature. Nayak's famous result on the impossibility of random access codes \cite{Nayak99} implies a linear lower bound for non-interactive protocols (ones that consists of only a single message from the server to the client), and implicitly, via extension of the same methods, also for multi-round protocols. Formal variants of this lower bound were proven also by Jain, Radhakrishnan and Sen~\cite{JRS09} (in terms of quantum mutual information) and by Baumeler and Broadbent \cite{BB15}. Indeed, one could trace back all of these results to the notion of \emph{adversary purification} which was used to show the impossibility of various cryptographic tasks in the information-theoretic quantum model starting as early as \cite{Lo97,LC97,M97}. In the context of QPIR, it can be shown that executing a QPIR protocol with sub-linear communication on a superposition of databases instead of on a single database, will leave the server at the end of the execution with a state that reveals some information about the index $i$. This is made explicit in \cite[Section 3.1]{JRS09} and is also implicit in the proof of \cite{BB15}. 

Most relevant to our work is the aforementioned \cite{BB15}, which provides an analysis from a cryptographic perspective and considers a well defined adversarial model known as privacy against specious adversaries, or \emph{the specious model} for short. This adversarial model was introduced by Dupuis, Nielsen and Salvail~\cite{DNS10} as a quantum counterpart to the classical notion of \emph{honest but curious} (a.k.a semi-honest) adversaries.\footnote{As \cite{DNS10} point out, their model is stronger, i.e.\ excludes a larger class of attacks, compared to the honest but curious model, even when restricted to a completely classical setting.} A specious adversary can be thought of as one that contains, as a part of its local state, a sub-state which is indistinguishable from that of the respective honest party, even when inspected jointly with the other party's local state.\footnote{More accurately, indistinguishability is required to hold even in the presence of an environment which can be arbitrary correlated (or entangled) with the parties' inputs. In the quantum setting this usually corresponds to the environment.} 

Let us provide a high level description of the specious model. We provide a general outline for two-party protocols, and not one that is specific to QPIR. Consider a protocol executed between parties $A,B$ on input registers $X,Y$ respectively. Let $A,B$ also denote the local state of the parties at a given point in time. Then the state of an honest execution of the protocol on inputs $XY$ can be described by the joint density matrix of the registers $XABY$. A specious adversarial strategy for party $A$ can be thought of as one where at any point in time, the local state of the adversary is of the form $A'XA$ (i.e.\ the adversary is allowed to maintain additional information, possibly in superposition with other parts of the system), such that the reduced density matrix of $XABY$ is still indistinguishable from the one obtained in an honest execution. This provides a potential advantage to a specious adversary (compared to an honest $A$) since it is quite possible that together with $A'$, the joint state is no longer honest. Thus the local view of the adversary, i.e.\ the registers $A'XA$, might in fact reveal information about $B$'s input $Y$ that was supposed to have been kept private. 

In the QPIR setting, say taking $A$ to be the server and $B$ to be the client, the register $X$ holds the database $\db$, and $Y$ holds the index $i$.
Indeed, \cite{BB15} shows that it is sufficient that $A'$ contains a purification of $XA$, where $X$ is a uniform distribution over all databases. We call this the {\it purification attack}. Thus, while the adversary pretends to execute the protocol on a randomly sampled database, it is in fact executed on a superposition of all possible databases at the same time (indeed this is the case since $A'$ contains a purification of $X$). As explained above, this methodology is not new, but \cite{BB15} analyze and show that no meaningful notion of QPIR can be achieved against this class of adversaries.

While the negative results could leave us pessimistic as to the abilities of quantum techniques to improve the state of the art on single-server PIR, there is some optimism suggested by two works. Le Gall~\cite{LeGall12} proposed a protocol with sub-linear communication (specifically $O(\sqrt{{\dblen}})$). Kerenidis et al.~\cite{KLGR16} proposed two protocols -- an explicit one, with $O(\log{\dblen})$ communication, which requires linear pre-shared entanglement; and a second protocol, with poly-logarithmic communication (and does not require pre-shared entanglement). In terms of privacy, it is shown that in a perfectly honest execution of the protocol, client's privacy is preserved. It might not be immediately clear how to translate this proof of privacy to the existing security models and reconcile it with the negative results. It is explained in \cite{LeGall12} that the protocol is not actually secure if the server deviates from the protocol. However, as \cite{BB15} observed, even a specious attacker that purifies the adversary can violate the security of the protocol, and the privacy proof strongly hinges on the honest execution using a classical database.

\paragraph{Challenges.} The state of affairs prior to this work, was that (non-trivial) QPIR was proven impossible even against fairly weak adversaries (namely, specious). Nevertheless, it appears that \cite{LeGall12,KLGR16} achieve some non-trivial privacy guarantee using sub-linear communication. This privacy guarantee appears not to be captured by the existing security model.  Lastly, we notice that all existing negative results are proven in a standalone model and did not consider protocols where the parties are allowed to share (honestly generated) setup information, such as the one by Kerenidis et al.~\cite{KLGR16}. In the quantum setting, a natural question is whether shared entanglement can help in achieving a stronger result.\footnote{We note that to the best of our understanding, even prior ``entropic'' results such as \cite{JRS09} seem to fall short of capturing the potential additional power of shared entanglement. This is essentially due to the property that if $AB$ are entangled, then it is possible that the reduced state of $B$ will have (much) higher von Neumann entropy than the joint $AB$ (whose entropy might even be $0$).} The goal of this work is to address these challenges.

\subsection{Our Results}
\paragraph{Anchored Privacy.} We start by formalizing a refinement of the standard notion of quantum privacy - one where the adversary is not allowed to purify its input register. We show anchored privacy against specious adversaries follows from anchored privacy against an honest party, if the protocol itself does not require parties to perform measurements (i.e.\ is measurement-free). 
Formally, using our notation from above, privacy in our model is only required to hold if the reduced density matrix of the register $X$ is a standard basis element, i.e.\ a fixed classical value. 
We call our model \emph{anchored} privacy as we can view our adversary as anchored to a specific value for its input $X$.

We observe that Le Gall's $O(\sqrt{n})$ protocol \cite{LeGall12} and the two protocols mentioned above by Kerenidis et al.~\cite{KLGR16} are in fact private against honest servers. We prove that explicitly for the pre-shared entanglement protocol by Kerinidis et al. in Appendix~\ref{sec:protocol}.  Using our reduction we can deduce that these protocol are also anchored private against specious adversaries, namely that so long as the adversary does not attempt to execute the protocol on a superposition of databases (and is still specious in the manner explained above), privacy is guaranteed.
In a sense, we formalize the folklore reliance on input purification to attack cryptographic schemes (and QPIR in particular), and show that in a model where input purification is impossible or prevented via some external restriction, it is possible to achieve security against specious adversaries. 

We believe this model is interesting for three main reasons:
\begin{enumerate}
    \item Conceptually, this model helps clarify the exact reason for the impossibility of QPIR - it is precisely because of the purification attack. Indeed, there is a formal sense in which some anchoring is necessary since we know that for any proposed protocol, allowing to execute on a superposition of inputs allows to violate security -- see the preceeding discussion in Section~\ref{sec:introduction}.
    
    \item We  view the anchored specious model as a stepping stone towards more robust notions. One intriguing future direction (mentioned briefly in our list of open problems) is to try to develop a malicious analog that still implements the ideology of ``forbidden input purification”, e.g. by forcing the adversary to ``classically open” the database before or after the execution in a manner that is consistent with the client’s output. Another interesting direction is to try to enforce anchoring using a two-server setting, thus achieving logarithmic two-server QPIR (which is currently still beyond reach).
    
    \item We believe that our new model may be plausible in certain situations where one could certify that the server cannot employ a superposition on databases.  We note that this model can be externally enforced, e.g.\ by conducting an inspection of the server's local computation device (with a very low probability) and making sure that it complies, and otherwise apply a heavy penalty. One could imagine such an inspection verifying that a copy of the database is stored on a macroscopic device that cannot be placed in superposition using available technology. Another example of a setting where the anchored model could be applicable is when the database contains information with some semantic meaning, so that the client can easily notice when a nonsense value has been used (this is somewhat similar to the setting considered in \cite{GLM08}). 
    We recall that semi-honest protocols are often used as building blocks, with additional external mechanisms that are employed to validate the assumptions of the model, and hope that our model can also be used in this way. Lastly, from a purely scientific perspective, we believe that formalizing and pinpointing a non-trivial model where non-trivial QPIR is possible will allow to better understand this primitive and the relation between quantum privacy and its classical counterpart. 
\end{enumerate}
    


\paragraph{Improved Lower Bound.} It would be instrumental to understand why the known QPIR lower bounds do not apply to our logarithmic protocol described above. Specifically, the protocol makes use of setup (pre-shared entanglement), and one could wonder whether this is the source of improvement, and perhaps with pre-shared entanglement it is possible to prove security even in the standard specious model. We show that this is not the case by providing a lower bound in the specious model even for the  \emph{one-sided communication} from the server to the client. Namely, we show that linear communication from the server to the client is necessary even if we allow arbitrary communication from the client to the server. In particular, this rules out the ability to use the setup to circumvent the lower bound, since the client (which is assumed to be honest) can generate the setup locally, and send the server's share across the channel at the beginning of the protocol. This completes the picture in terms of the impossibility of QPIR in the specious model and further justifies our relaxation of the model in order to achieve meaningful results. 

Noting that PIR can be thought of as a special case of Fully Homomorphic Encryption (FHE), our lower bound implies that even a Quantum Fully Homomorphic Encryption (QFHE) with (even approximate) information theoretic security cannot have non-trivial communication complexity, even if the QFHE protocol is allowed to make use of shared entanglement between the server and the client. We thus generalize (to allow shared prior entanglement) the impossibility results for (even imperfect) QPIR of \cite{BB15} (as well as those of \cite{YPF14} which explicitly referred to QFHE).

\subsection{Overview of Our Techniques}

\paragraph{Anchored-Specious Security.} Recall the notation introduced above for two party protocol $(A,B)$ on inputs $(X,Y)$, and recall that a specious adversary can be thought of as one where the local state of the adversary is of the form $A'XA$. Now let us consider the case of measurement-free protocols and also assume that the client's input $Y$ is a pure state (this can be justified since otherwise we can apply our argument on the joint state of $Y$ and its purifying environment instead of $Y$ itself). In such an execution, it holds that at any stage $XABY$ is a pure superposition (i.e.\ its density matrix is of rank $1$). Now let us consider the joint state together with the specious adversary's additional register, i.e.\ $A'XABY$. Since $(XABY)$ is pure, $A'$ cannot be entangled with it, and therefore $A'$ is in tensor product with the remainder of the state, namely $(XABY)$. It follows that the status of the register $A'$ can be simulated at any point in time without any knowledge of the other components of the protocol. There is a delicate point here, since $A'$ may indeed be in tensor product, but we must also argue that it is independent of $Y$. Intuitively, to see why such dependence on $Y$ cannot occur consider, e.g., $Y=\ket{y_1}+\ket{y_2}$. Then $YA'$ is in the sate $\ket{y_1}\otimes \rho_{A'}+\ket{y_2}\otimes \rho_{A'}$ (importantly the same $\rho_{A'}$ appears twice). However, this state is exactly the purification of executing the protocol either with $Y=\ket{y_1}$ or with $Y=\ket{y_2}$. We conclude that $\rho_{A'}$ must be the same in both settings, and by extension it can be shown to be the same for all $Y$.

After taking care of $A'$, we need to consider the other part of the adversary's state, namely the register $(XA)$. This register is, by definition, identical (or indistinguishable) from the state of an honest party during the execution. Recall that we assume our protocol is anchored private against honest servers. So the local honest state $(XA)$ is guaranteed not to leak information about $B$'s input. Add to that the conclusion about $A'$ being in tensor product and independent of $B$'s state, and we get that the entire local state of the specious adversary does not reveal any disallowed information.

As a conclusion, since we can show, e.g.\ in Le Gall's protocol or in our logarithmic protocol, that an honest execution with a classical $X$ does not leak information about $Y$, this will also be the case in the anchored-specious setting. 


Obviously many details are omitted from this high level overview. For example, a specious adversary is not required to make $(XABY)$ identical to an honest execution but rather only statistically close (in trace distance), which requires a more delicate analysis. Furthermore, the formal construction of a simulator for the adversary as required by the specious definition requires some care to detail.
For the formal definitions and analysis see Section~\ref{sec:model} below.

\paragraph{Our Lower Bound.} We first note that previous lower bound proofs in~\cite{Nayak99,BB15} bounded the \emph{total} communication complexity by a reduction to quantum random access codes. It is not a-priori clear how to generalize this proof method to the presence of shared entanglement. To do so, we provide a new lower bound argument that establishes a linear lower bound on the \emph{server's} communication complexity. Specifically, we show that the server needs to transmit at least roughly $n/2$ qubits to the client, no matter how many qubits is transmitted from the client to the server (assuming that the protocol has sufficiently small correctness and privacy error). As we mentioned above, such a lower bound trivially extends to hold with prior shared entanglement, since one can think of that the shared entanglement is established by the client sending messages to the server.

Our new lower bound argument is based on an interactive leakage chain rule in~\cite{LC18c} and might even be considered conceptually simpler than previous methods. At a high level, we consider a server holding a uniformly random database $\vc{a} \in \zo^n$ and running a QPIR protocol with a client. Initially, from the client's point of view, the database $\vc{a}$ has $n$-bits of min-entropy, and the protocol execution can be viewed as an ``interactive leakage'' that leaks information about $\vc{a}$ to the client. Let $m_A$ and $m_B$ denote the server and the client's communication complexity in the protocol. The interactive leakage chain rule in~\cite{LC18c} states that the min-entropy of $\vc{a}$ can only be decreased by at most $\min\{ 2m_A, m_A + m_B\}$. More precisely, let $\rho_{AB}$ denote the states at the end of the protocol execution where the $A$ register stores the (classical) random database $\vc{a}$ and $B$ denotes the client's local register. The interactive leakage chain rule states that $H_{\min}(A|B)_\rho \geq n - \min\{ 2m_A, m_A + m_B\}$. By the operational meaning of quantum min-entropy, given the client's state $\rho_B$, one cannot predict the database correctly with probability higher than $2^{-(n - \min\{ 2m_A, m_A + m_B\})}$. On the other hand, suppose the protocol is secure against specious servers with sufficiently small correctness and privacy error. We can combine the by-now standard lower bound argument by Lo~\cite{Lo97} and gentle measurement\cite{Win99,Aar04,ON07}, we show that one can reconstruct the database $\vc{a}$ from the client's state $\rho_B$ with a constant probability. 
Combining both claims allows us to establish lower bounds on both the server's and the total communication complexity in a unified way.

\subsection{Remaining Open Problems}

We proposed a new model and a new protocol which, we believe, resurfaces the question of what can be achieved in the context of QPIR. We believe that a number of intriguing questions still remain for future work.

\begin{enumerate}
	\item As discussed above, our model is a relaxation of the specious model, which is by itself a semi-honest model. Such models are fairly restrictive in the sense that they make structural assumptions on the adversary (i.e.\ that it follows the protocol, or contains a part that follows the protocol). Obviously, if we hope for non-trivial results, any model that we formalize must preclude purification of input. It is thus an intriguing question whether it is possible to formulate \emph{malicious} adversarial models that are still purification-free, and what can be said about the plausibility of QPIR in such models.	
	The current definition of anchored privacy will need to be amended, since a malicious server is allowed to just ignore its prescribed input, so a different method of anchoring needs to be devised.
	
	\item Another natural question is whether setup is necessary to achieve logarithmic QPIR in the anchored specious model. We know from Kerenidis et al.'s result that polylogarithmic communication is achievable even without setup. Is there a reason can only improve it when assuming a setup? Another surprising aspect is that the shared entanglement created during the setup is not consumed during the protocol, and can be used for other needs after the execution of the protocol (e.g., running another execution of PIR, or teleportation). A similar phenomenon occurs in quantum information: catalyst quantum states are useful for mapping one bi-partite state to another using LOCC, without consuming the catalyst state~\cite{DP99,Kli07}. The related notion of quantum embezzlement~\cite{vDH03} has a similar property, but in this case, the original shared state changes slightly. The authors are not aware of any other cryptographic protocol with this non-consumption property. 
	
	\item Most state of the art classical PIR protocols (both in the multi-server setting and in the computational cryptographic setting) only require one round of communication. That is, one message (query) from the client to the server (or servers) and one response message. All the existing sublinear QPIR  protocols have multiple rounds.   Understanding the round complexity of QPIR in light of the classical state of the art is also an intriguing direction.
	
	\item A main contribution of this work is to formalize the notion of anchored security and show it can be used to provide a non-trivial 
cryptographic primitive. It would be interesting to study the relevance of this notion (or adequately adapted versions) in the context of a variety of other cryptographic tasks. In particular, the question of whether it is possible to construct information theoretically secure fully homomorphic encryption (FHE) given quantum channels has received attention in recent years (see, e.g., \cite{YPF14}). In homomorphic encryption, the server has a function $f$ and the client has an input $x$, and the goal of the protocol is for the client to learn $f(x)$ without revealing any information about $x$. PIR and FHE functionalities are intimately related (think about a function $f_{\db}(i)=\db[i]$ for FHE, and about executing PIR with database equal to the truth table of some function), and it is thus intriguing whether the anchored model is applicable in the context of FHE as well.
\end{enumerate}

\subsection{Paper Organization}

General preliminaries are provided in Section~\ref{sec:prelim}. We present our new model, and the proof that for pure protocols honest security implies anchored specious security in Section~\ref{sec:model}. Our new lower bound is stated and proven in Section~\ref{sec:lb}. In Appendix~\ref{sec:protocol}, we show that the protocol by Kerenidis et al. is anchored private against specious adversaries.

%% file: prelim.tex

\def\tdist{\Delta}
\newcommand{\ketbra}[1]{\ket{{#1}}\bra{{#1}}}

\section{Preliminaries}
\label{sec:prelim}

Standard preliminaries regarding Hilbert spaces and quantum states can be found in Appendix~\ref{sec:stdprelim}. We provide below background and definitions concerning two-party quantum protocols, specious adversaries and quantum private information retrieval.

\subsection{Two-Party Quantum Protocols}

As in \cite{BB15}, we base our definitions on the works of \cite{GW07} and \cite{DNS10}. However, we make slight adaptations to allow for prior entanglement between the parties.

\begin{definition}[Two-party quantum protocol]\label{def:qprot} An \textit{s-round, two-party quantum protocol},  denoted $\Pi=\{\mathscr{A,B},\rho_{joint},s\}$ consists of:
\begin{enumerate}
\item input spaces $\mathcal{A}_0$ and $\mathcal{B}_0$ for parties $\mathscr{A}$ and $\mathscr{B}$ respectively,
\item initial spaces $\mathcal{A}_p$ and $\mathcal{B}_p$ ($p$ for pre-shared state) for parties $\mathscr{A}$ and $\mathscr{B}$ respectively,
\item a joint initial state $\rho_{joint}\in\mathcal{A}_p\otimes\mathcal{B}_p$, split between the two parties,
\item memory spaces $\mathcal{A}_1,\dots,\mathcal{A}_s$ for $\mathscr{A}$ and $\mathcal{B}_1,\dots,\mathcal{B}_s$ for $\mathscr{B}$, and communication spaces $\mathcal{X}_1,\dots,\mathcal{X}_s$, $\mathcal{Y}_1,\dots,\mathcal{Y}_{s-1}$,
\item an $s$-tuple of quantum operations $(\mathscr{A}_1,\dots,\mathscr{A}_s)$ for $\mathscr{A}$, where $\mathscr{A}_1:L(\mathcal{A}_0\otimes\mathcal{A}_p)\mapsto L(\mathcal{A}_1\otimes\mathcal{X}_1)$, and $\mathscr{A}_t:L(\mathcal{A}_{t-1}\otimes\mathcal{Y}_{t-1})\mapsto L(\mathcal{A}_t\otimes\mathcal{X}_t)$ $(2\leq t \leq s)$,
\item an $s$-tuple of quantum operations $(\mathscr{B}_1,\dots,\mathscr{B}_s)$ for $\mathscr{B}$, where $\mathscr{B}_1:L(\mathcal{B}_{0}\otimes\mathcal{B}_p\otimes\mathcal{X}_{1})\mapsto L(\mathcal{B}_1\otimes\mathcal{Y}_1)$, $\mathscr{B}_t:L(\mathcal{B}_{t-1}\otimes\mathcal{X}_{t})\mapsto L(\mathcal{B}_t\otimes\mathcal{Y}_t)$ $(2\leq t \leq s-1)$, and $\mathscr{B}_s:L(\mathcal{B}_{s-1}\otimes\mathcal{X}_{s})\mapsto L(\mathcal{B}_s)$.
\end{enumerate}
\end{definition}

Note that in order to execute a protocol as defined above, one has to specify the input, namely a quantum state $\rho_{in}\in S(\mathcal{A}_0\otimes \mathcal{B}_0)$ from which the execution starts.

\begin{definition}[Protocol Execution]
If $\Pi=\{\mathscr{A},\mathscr{B}, \rho_{joint},s\}$ is an $s$-round two-party protocol, then the state after the $t$-th step $(1\leq t \leq 2s)$, and upon input state $\rho_{in}\in S(\mathcal{A}_0\otimes \mathcal{B}_0\otimes\cR)$, for any $\cR$, is defined as
\[
\rho_t(\rho_{in}):=(\mathscr{A}_{(t+1)/2}\otimes I_{\mathcal{B}_{(t-1)/2}}) \dots (\mathscr{B}_1\otimes I_{\mathcal{A}_1})(\mathscr{A}_1\otimes I_{\mathcal{B}_0,\mathcal{B}_p})(\rho_{in}\otimes\rho_{joint}),
\]
for $t$ odd, and 
\[
\rho_t(\rho_{in}):=(\mathscr{B}_{t/2}\otimes I_{\mathcal{A}_{t/2}}) \dots (\mathscr{B}_1\otimes I_{\mathcal{A}_1})(\mathscr{A}_1\otimes I_{\mathcal{B}_0,\mathcal{B}_{p}})(\rho_{in}\otimes\rho_{joint}),
\]
for $t$ even.
We define the final state of protocol $\Pi=\{\mathscr{A},\mathscr{B}, \rho_{joint},s\}$ upon input state $\rho_{in}\in S(\mathcal{A}_0\otimes \mathcal{B}_0\otimes\cR)$ as: $\left[\mathscr{A}\circledast\mathscr{B}\right](\rho_{in}):=\rho_{2s}(\rho_{in})$.
\end{definition}

The communication complexity of a protocol is the number of qubits that are exchanged between the parties. Slightly more generally, we can consider the logarithm of the dimension of the message registers $\cX_t$, $\cY_t$. The formal definition thus follows.

\begin{definition}[Communication Complexity]
The  \emph{communication complexity} of a protocol as in Definition~\ref{def:qprot} is $$\sum_{t=1}^s \log \dim(\cX_t)+\sum_{t=1}^{s-1}\log \dim(\cY_t)~.$$
We sometimes also refer to one-sided communication complexity, i.e.\ the total communication originating from one party to the other. The communication complexity of $\mathscr{A}$ is defined to be the communication originating from $\mathscr{A}$, or formally $\sum_{t=1}^s \log \dim(\cX_t)$. Symmetrically the communication complexity of $\mathscr{B}$ is $\sum_{t=1}^{s-1} \log \dim(\cY_t)$.
\end{definition}

\subsection{Specious Adversary}
Given a two-party quantum protocol $\Pi=\{\mathscr{A,B},\rho_{joint},s\}$,
an adversary $\tilde{\mathscr{A}}$ for $\mathscr{A}$ is an $s$-tuple of quantum operations $(\mathscr{A}_1,\dots,\mathscr{A}_s)$, where
$\tilde{\mathscr{A}}_1: L(\tilde{\mathcal{A}}_{0})\mapsto L(\mathcal{A}_{1}\otimes \cX_1)$ and
$\tilde{\mathscr{A}}_t: L(\tilde{\mathcal{A}}_{t-1}\otimes \mathcal{Y}_{t-1} )\mapsto L(\tilde{\mathcal{A}}_{t}\otimes \cX_t)$, $2\leq t\leq s$, with $\tilde{\cA}_1,\dots, \tilde{\cA}_s$ being $\tilde{\mathscr{A}}$'s memory spaces.
The global state after the $t$th step of a protocol run with $\tilde{\mathscr{A}}$ is $\tilde{\rho}_t(\tilde{\mathscr{A}},\rho_{in})$.
An adversary $\tilde{\mathscr{B}}$ for $\mathscr{B}$ is similarly defined.

\begin{definition}[Specious adversaries]\label{def:speciousadv} Let $\Pi=\{\mathscr{A},\mathscr{B}, \rho_{joint},s\}$ be an $s$-round two-party protocol. An adversary $\tilde{\mathscr{A}}$ for $\mathscr{A}$ is said to be $\gamma$-specious, if there exists a sequence of quantum operations (called recovery operators) $\mathscr{F}_1,\dots,\mathscr{F}_{2s}$, such that for $1\leq t\leq 2s$ and for all $\rho_{in}\in S(\mathcal{A}_0\otimes\mathcal{B}_0\otimes\mathcal{R})$:
\begin{enumerate}
\item For all $t$ even, $\mathscr{F}_t: L(\tilde{\mathcal{A}}_{t/2})\mapsto L(\mathcal{A}_{t/2})$.
\item For all $t$ odd, $\mathscr{F}_t: L(\tilde{\mathcal{A}}_{(t+1)/2}\otimes \mathcal{X}_{(t+1)/2} )\mapsto L(\mathcal{A}_{(t+1)/2} \otimes \mathcal{X}_{(t+1)/2})$.
\item For every input state $\rho_{in}\in S(\mathcal{A}_0\otimes\mathcal{B}_0\otimes \mathcal{R})$, for any $\cR$,  
\begin{equation}
\Delta\left( \left(\mathscr{F}_t\otimes I_{\mathcal{B}_t,\mathcal{R}} \right)\left(\tilde{\rho}_t(\tilde{\mathscr{A}},\rho_{in})\right), \rho_t\left(\rho_{in} \right)  \right) \leq \gamma.
\label{eq:speciousness}
\end{equation}

\end{enumerate}
\end{definition}
A $\gamma$-specious adversary $\tilde{\mathscr{B}}$ for $\mathscr{B}$ is similarly defined.

\subsection{Quantum Private Information Retrieval} \label{sec:QPIR}

We define QPIR similarly to \cite{BB15}.

\begin{definition}[Quantum Private Information Retrieval]\label{def:qpir}
An $s$-round, $n$-bit Quantum Private Information Retrieval protocol (QPIR) is a two-party protocol $\Pi_{QPIR}=\{\mathscr{A},\mathscr{B},\rho_{joint},s\}$, where $\mathscr{A}$ is the server, $\mathscr{B}$ is the client, and $\rho_{joint}$ is an initial state shared between them prior to the protocol.
We call $\Pi_{QPIR}$ $(1-\delta)$-correct if, for all inputs $\rho_{in}=\ket{x}\bra{x}_{\cA_0}\otimes \ket{i}\bra{i}_{\cB_0}$, with $x=x_1,\dots, x_n\in\{0,1\}^n$ and $i\in\{1,\dots,n\}$, there exists a measurement $\mathcal{M}$ acting on $\cB_s$
with outcome $0$ or $1$, such that:
\begin{align*}
& \Pr\left\{ \cM\left( \tr_{\cA_s} \left[\mathscr{A}\circledast\mathscr{B}\right](\rho_{in}) \right)=x_i \right\}\geq 1-\delta~.
\end{align*}

If $\delta=0$ we say that the protocol is perfectly correct.

We call $\Pi_{QPIR}$ $\epsilon$-private  
against a (possibly adversarial) server $\tilde{\mathscr{A}}$, if there exists a 
sequence of quantum operations (simulators) $\mathscr{I}_1,\dots,\mathscr{I}_{s-1}$, where $\mathscr{I}_{t}: L(\cA_0\otimes\cA_p)\mapsto L(\tilde{\cA}_t \otimes \cY_t)$, such that for all $1\leq t\leq s-1$ and for all $\rho_{in}\in S(\cA_0\otimes \cB_0\otimes \cR)$,
\begin{align}
& \Delta\left(    \tr_{\cB_0} \left(\mathscr{I}_{t} \otimes \mathbb{I}_{\cB_0,\cR} (\rho_{in})\right) 
,    \tr_{\cB_t}(\tilde{\rho}_{2t}(\tilde{\mathscr{A}},\rho_{in}) ) \right) \leq \epsilon~.
\label{eq:requirements_from_QPIR}
\end{align}
If $\epsilon=0$ we say that the protocol is perfectly private.

We say that a QPIR protocol is $\epsilon$-private against a class of servers if it is $\epsilon$-private against any server from this class.
\end{definition}

We note that in the above definition privacy is required to hold also for adversarial input states for the client and server, which also includes inputs in superposition, and even for the case where the client and server (and possibly a third party) are entangled.
Nayak~\cite{Nayak99,ANTV02} showed that a perfectly private QPIR protocol, even only against $0$-specious servers, must have communication complexity at least $(1-H(1-\delta))n$, where $H(p)$ is the binary entropy function. 
Baumeler and Broadbent~\cite{BB15} extended  this lower bound to the case of $\epsilon>0$ and presented a communication lower bound of
\begin{align}
\left(1-H\left(1-\delta-2\sqrt{\epsilon(2-\epsilon)}\right)\right)n~. \label{eq:nayak_bound}
\end{align}

%% file: model.tex

\section{Anchored Privacy Against Specious Adversaries}
\label{sec:model}
We now present our new restricted notion of privacy, that we call \emph{anchored privacy}. A protocol is anchored private if it satisfies the standard definition of privacy with respect to classical inputs on the adversary's side. There is no privacy requirement for superposition input states on the adversary's side (and therefore this notion of privacy is weaker, and hence, easier to achieve). A formal definition follows.

\begin{definition}[Anchored Privacy]
A QPIR protocol is anchored $\epsilon$-private if Eq.~\eqref{eq:requirements_from_QPIR} holds for all $\rho_{in}\in \cA_0\otimes\cB_0\otimes\cR$ (for any $\cR$), for which $\rho_{in}|_{\cA_0}=\ket{x}\bra{x}$ for some $x\in\{0,1\}^n$.
\end{definition}

We note that prior intuitive notions of security such as that implied by the analysis of Le Gall \cite{LeGall12} in fact correspond to anchored privacy against honest servers. Our main theorem below shows that this type of privacy extends to the specious setting as well.

\begin{theorem}\label{thm:main} Let $\Pi$ be a measurement-free QPIR protocol which is anchored $\epsilon$-private against honest servers, then $\Pi$ is anchored $({\epsilon}+3\sqrt{2\gamma})$-private against $\gamma$-specious servers.
\end{theorem}

Critically, the theorem only holds for measurement-free QPIR protocols. To see this, consider the following protocol, which will be anchored-private against honest servers but not anchored-private against specious ones. Let $\Pi$ be a QPIR protocol which is anchored-private against honest servers (e.g., Le-Gall's protocol \cite{LeGall12}). Now consider the following protocol $\Pi'$ which first generates a superposition over all possible databases, then measures this superposition to obtain a classical value for the database. It then runs $\Pi$ on this measured database (with the client using its real input index). Finally, both parties toss out the output of this first execution, and run $\Pi$ again, now using the actual input database.

Let us first see that $\Pi'$ is anchored-private against honest servers. 
This follows since $\Pi$ is secure against honest adversaries when executed over input states in which the server's input is classical, and hence so is $\Pi'$ which just consists of two sequential executions of $\Pi$ over classical databases. However, a purification of an honest server allows to execute a purification attack 
on the first execution of $\Pi$ in a way that allows to recover the client's input, even though the database used as input for $\Pi'$ is completely classical.


\paragraph{Warm-up.} We first give a proof under some simplifying assumptions: (i) $\gamma=\epsilon=0$. (ii) the input is pure (iii) the purification space is trivial: $\cR=\mathbb C$ and (iv) the specious server's quantum operations $\tilde{\mathscr{A}}_t$ are unitary. The main point that makes the analysis easier in this case is assumption (i). 

Fix a step of the protocol $t$. 
\begin{enumerate}
\item We claim that for every unitary $\gamma$-specious adversary, which is perfect (i.e. $\gamma=0$) the entire state, (written in some {\it fixed} but maybe non standard basis), is of the form $\ket{\eta}_{\cS'} \otimes \ket{\psi_t}_{\cS,\cC}$ where $\ket{\psi_t}$ is the state that an honest server and client would have when running on the same input. Here, and later, we use the notation $\cS$ for all of the honest server registers at step $t$, $\cC$ for all of the client's registers at step $t$ and $\cS'$ for the specious server's ancillary register at step $t$.  Crucially, $\ket{\eta}$ is independent of the (server and client) input.

We now prove the above claim. By the specious property, we know that there exists a quantum operation $\mathscr{F}_t$ which maps the global state at the $t$th stage to the state $\ket{\psi_t}$. We know that the state in step $t$ in the honest run is necessarily pure since $\Pi$ is measurement free. W.l.o.g. we can assume that the operation $\mathscr{F}_t$ is a unitary $U_t$, followed by tracing out everything other then the $\cS$ and $\cC$ registers.

Let's assume towards contradiction that the state in the basis $U_t^\dagger$ is of the form $\ket{\eta (input)}\otimes \ket{\psi_t}$, where $\ket{\eta(input)}$ depends on the input (where here we mean both the client and the server's input). There must be two different input states such that running them would give $\ket{\eta(1)}\otimes \ket{\psi_t(1)}$ and $\ket{\eta(2)}\otimes \ket{\psi_t(2)}$ for which $\ket{\eta(1)}\neq \ket{\eta(2)}$. Since the honest runs are entirely unitary (by the measurement-free property) and have different inputs, necessarily,  $\ket{\psi_t(1)}\neq \ket{\psi_t(2)}$. By running the specious adversary on a superposition of these two inputs, we get that after applying $\mathscr{F}_t$, the state becomes a mixture of the two states, $\ket{\psi_t(1)}$ and $\ket{\psi_t(2)}$. This contradicts the perfect specious property (see Eq.~\eqref{eq:speciousness}) -- which requires the state to be the pure (since all the operations of the client and honest servers are unitary, and their input in this case is pure).

\item By the perfect anchored-privacy against the honest server, the state $\rho_{t}=\tr_C(\ketbra{\psi_t}_{S,C})$ is independent of the client's input, and therefore, could only depend on $x$ -- the server's input. To emphasize that independence on the client's input (and possible dependence on the server's input), we denote the state $\rho_{t}$ by $\rho_{t}(x)$. 
\end{enumerate}

Our goal is to show the anchored-privacy property for the specious server. Indeed, the two points above show that the specious server's state (in the fixed basis we choose to work in) is $\ketbra{\eta}\otimes \rho_t(x)$, which is independent of the client's input. Therefore the simulator can generate that state exactly by using the server's classical input $x$, as required (see Eq.~\eqref{eq:requirements_from_QPIR}).  

\paragraph{Outline of the general proof.} For each round $t$ we construct a simulator for the server in the following way: we first construct a simulator $\tilde{\sI}_t^{x_0,0}$ for input $\ket{x_0}\otimes \ket{0}$ where $\ket{x_0}$ is an input for the server and $\ket{0}$ is an input for the client. We construct this simulator using the simulator for the honest server along with the 'specious operator', and an ancillary state $\ket{\sigma_{x_0,0}}$. We then show that  $\ket{\sigma_{x_0,0}}$ is also an appropriate ancillary state for any input $\ket{x}\otimes \ket{\eta}$. Using this, we show that $\tilde{\sI}_t^{x,0}$ is indeed a simulator for any input $\ket{x}\otimes \ket{\eta}$, with slightly worse parameters.
 
We are now ready to give the proof in full generality:
\begin{proof}[Theorem~\ref{thm:main} (Proof)]
Let $\Pi$ be a purified QPIR protocol which is anchored $\epsilon$-private against honest servers, and let $\tilde{\sA}$ be a $\gamma$-specious server for $\Pi$. W.l.o.g we can assume that $\tilde{\sA}$ is purified, namely, a unitary\footnote{This is because we can include the purification register at any point, as the server could have included himself rather than throwing it away}. From now on, we will fix $t$. We can denote 

\begin{align}\label{psi i}
\ket{\psi_t^{\rho_{in}}}\bra{\psi_t^{\rho_{in}}}= \rho_t(\rho_{in})
\end{align}

where $\ket{\psi_t^{\rho_{in}}}\in \cS \otimes \cC \otimes \cR$ for some $\cR$, and we use $\cS$ to represent the server's registers $\cS=\cA_t\otimes \cY_t \otimes \cA_{p}$ (for $t$ odd. otherwise $\cS=\cA_t\otimes\cA_{p}$), 
and $\cC$ to represent the client's registers $\cC=\cB_t\otimes \cX_t \otimes \cB_{p}$ (for $t$ even. otherwise $\cC=\cB_t\otimes \cB_{p}$). Furthermore, w.l.o.g we assume the various recovery operators for $\tilde{\sA}$ are purified. That is, there exist unitary operators $\hat{\sF}_t$ such that $\sF_t(\cdot)=\tr_{\cS'}\left(\hat{\sF}_t(\cdot)\right)$ for some purification space $\cS'$ which is at the hands of the server (from now on, for the sake of this proof, where we say "recovery operators" we regard these unitary $\hat{\sF}_t$ operators). Therefore we can denote
\begin{align}\label{tilde psi i}
\ket{\tilde{\psi}_t^{\rho_{in}}}\bra{\tilde{\psi}_t^{\rho_{in}}}=\tilde{\rho}_t\left(\tilde{\sA}, \rho_{in} \right)
\end{align}
where w.l.o.g $\ket{\tilde{\psi}_t^{\rho_{in}}}\in S' \otimes S \otimes C \otimes R$.
We note that all of the unitary operators - $\sA_t, \sB_t $ which are used in the original protocol (by either the server or the client), $\tilde{\sA}_t$ which are used by the specious server $\tilde{\sA}$, and the recovery $\hat{\sF}_t$ operators are independent of both the client's and the server's inputs 

For each round $t$, we will start by constructing a simulator for $\tilde{\sA}$ acting on $\rho_{in}= \ket{x_0}\bra{x_0}_{\cA_0}\otimes \ket{0}\bra{0}_{\cB_0}$, where $x_0\in\{0,1\}^n$ (in this specific input, $\cR$ is trivial and is thus omitted). 
By $\gamma$-speciousness of $\tilde{\sA}$, along with our purification assumptions, there exists a unitary recovery operator $\hat{\sF}_{2t}:L(\tilde{\cA}_t)\mapsto L (\cS'\otimes\cA_t )$ such that
\begin{align}
\Delta\left( 
\tr_{\cS'}\left( \left( \hat{\sF}_{2t} \otimes \bbI_{\cC}\right) \ket{\tilde{\psi}_{2t}^{ \ket{x_0}\otimes\ket{0} }} \right),
\ket{\psi_{2t}^{ \ket{x_0}\otimes\ket{0}  }}
\right)\leq\gamma
\end{align}
By Lemma \ref{trace-in lemma}, this means that there exists a state $\ket{\sigma_{x_0,0}}\in \cS'$ such that:
\begin{align}\label{speciousness for x,0}
\Delta\left(   \left(\hat{\sF}_{2t}\otimes \bbI\right) \ket{\tilde{\psi}_{2t}^{\ket{x_0}\otimes \ket{0}}}   , \ket{\sigma_{x_0,0}} \otimes  \ket{\psi_{2t}^{\ket{x_0}\otimes \ket{0}}}
\right)\leq \sqrt{\gamma}
\end{align}
We can now operate on Eq.~\eqref{speciousness for x,0} with $\hat{\sF}_{2t}^\dagger \otimes \bbI$ to get:

\begin{align}\label{0 simulator part1}
\Delta\left(    
\ket{\tilde{\psi}_{2t}^{\ket{x_0}\otimes \ket{0}}}  ,
\left( \hat{\sF}_{2t}^\dagger \otimes \bbI \right)    \left(\ket{\sigma_{x_0,0}} \otimes   \ket{\psi_{2t}^{\ket{x_0}\otimes \ket{0}}} \right)
\right)\leq \sqrt{\gamma}
\end{align}
The above connects the states derived from the execution with the 
specious server to that with the honest server.  
By anchored $\epsilon$-privacy of $\Pi$ against honest servers, there exists a simulator $\sI_t:L(\mathcal{A}_0\otimes\mathcal{A}_{p})\mapsto L(\cA_t\otimes \cX_t)$ such that for all $x\in\{0,1\}^n$ and $\ket{\alpha}\in \cB_0\otimes\cR$, for any $\cR$, 
\begin{align}\label{0 simulator part2}
\Delta\left( 
\tr_{\cB_0,\cB_{p}}
\left(
\left(\sI_t\otimes\bbI_{\cB_0,\cB_{p}}
     \right) \circ
  \left(
  \ket{x}\bra{x}_{\cA_0}
  \otimes\ket{\alpha}\bra{\alpha}_{\cR,\cB_0}
  \otimes \rho_{joint} 
  \right) 
  \right), \tr_{\cB_t}\left(\ket{\psi_{2t}^{\ket{x}\otimes\ket{\alpha}}} \right)
\right)\leq\epsilon
\end{align}
(In fact, the above holds for any 
mixture over such $\alpha$'s, by 
convexity). 
We can now define the simulator for $\rho_{in}$ corresponding to input state $\ket{x_0}\otimes\ket{0}$ to be the following unitary embedding from $\cA_0\otimes\cA_p$ to $\cS'\otimes\cA_0\otimes\cA_p$: 
\begin{align}\label{0 adversarial simulator}
\tilde{\sI}_t^{x_0, 0}(\cdot) =
  \hat{\sF}_{2t}^\dagger \circ 
  \left( \ket{\sigma_{x_0,0}}\bra{\sigma_{x_0,0}} \otimes  \sI_t \left(\cdot \right) \right)
\end{align}

To show that it indeed satisfies the requirements 
from a simulator, we combine Eqs.~\eqref{0 simulator part1},\eqref{0 adversarial simulator}, and~\eqref{0 simulator part2} for $x=x_0,\ \ket{\alpha}=\ket{0}$, to get that 

\begin{align}\label{0 simulator final}
\Delta\left(
tr_{\cB_0,\cB_{p}}
\left(
\left(\tilde{\sI}_t^{x_0,0}\otimes\bbI_{\cB_0,\cB_{p}}
     \right) \circ
  \left(
  \ket{x_0}\bra{x_0}_{\cA_0}
  \otimes\ket{0}\bra{0}_{\cB_0}
  \otimes \rho_{joint} 
  \right) 
  \right),
\tr_{\cB_t}\left(  \ket{\tilde{\psi}_{2t}^{\ket{x_0}\otimes \ket{0}}}  \right)
\right) \leq \epsilon + \sqrt{\gamma}
\end{align}

We now define the simulator for any input to be this exact simulator: 
\begin{align}\label{adversarial simulator}
\tilde{\sI}_t(\cdot) = \tilde{\sI}_t^{x_0, 0}(\cdot);
\end{align} 
In the remainder of the proof we show that $\tilde{\sI}_t(\cdot)$ satisfies an inequality similar to Eq.~\eqref{0 simulator final} with respect to all classical server inputs $x\in \{0,1\}^n$  
(not necessarily $x_0$)
and any input state $\ket{\alpha}\in \cB_0\otimes\cR$ for any $\cR$, as well as for a mixture of such $\alpha$'s; this would imply anchored privacy for the specious server. 
To this end we show that also for this input, a similar inequality to Eq.~\eqref{0 simulator final} holds (with a slightly worse bound).  
Define \[\ket{x\alpha_+} = \frac{1}{\sqrt{2}}\ket{0}_{\cR'}\ket{x_0}_{\cA_0}\ket{0}_{\cB_0,\cR} + \frac{1}{\sqrt{2}}\ket{1}_{\cR'}\ket{x}_{\cA_0}\ket{\alpha}_{\cB_0,\cR},\] where 
we have added an additional (control) qubit in the space $\cR'$. 
The specious adversary condition applies to this input state as well, and thus using the same derivation as for Eq.~\eqref{0 simulator part1}) we get: 
\begin{align}\label{y+ simulator part1}
\Delta\left(    
\ket{\tilde{\psi}_{2t}^{\ket{x\alpha_+}}}  ,
\left( \hat{\sF}_{2t}^\dagger \otimes \bbI \right)    \left(\ket{\sigma_{x\alpha_+}} \otimes   \ket{\psi_{2t}^{\ket{x\alpha_+}}} \right)
\right)\leq \sqrt{\gamma}
\end{align}

Using the fact that neither the server nor the client act on the $\cR'$ register, we get:

\begin{align} \label{superposed psi i}
\ket{\psi_{2t}^{\ket{x\alpha_+}}} 
= \frac{1}{\sqrt{2}}
\ket{0}_{\cR'}\otimes\ket{\psi_{2t}^{\ket{x_0}\otimes \ket{0}}}_{\cS,\cC,\cR}   + 
\frac{1}{\sqrt{2}}\ket{1}_{\cR'}\otimes \ket{\psi_{2t}^{\ket{x}\otimes \ket{\alpha}}}_{\cS,\cC,\cR}
\end{align}
Similarly, since the same is true for the adversarial run, we get:

\begin{align}\label{superposed tilde psi i}
\ket{\tilde{\psi}_{2t}^{\ket{x\alpha_+}}}
= 
\frac{1}{\sqrt{2}}
\ket{0}_{\cR'}\otimes\ket{\tilde{\psi}_{2t}^{\ket{x_0}\otimes \ket{0}}}_{\cS,\cC,\cR}
+   \frac{1}{\sqrt{2}}\ket{1}_{\cR'}\otimes \ket{\tilde{\psi}_{2t}^{\ket{x}\otimes \ket{\alpha}}}_{\cS',\cS,\cC,\cR}
\end{align}

We plug Eqs.~\eqref{superposed psi i} and \eqref{superposed tilde psi i} into Eq.~\eqref{y+ simulator part1}, and project the register $\cR'$ in the resulting state onto $\ket{1}_{\cR'}$ to get:

\begin{align}\label{y+ simulator part2}
\Delta\left(    
\frac{1}{\sqrt{2}}\ket{1}_{\cR'}\otimes \ket{\tilde{\psi}_{2t}^{\ket{x}\otimes \ket{\alpha}}}_{\cS,\cC} ,
\left( \hat{\sF}_{2t}^\dagger \otimes \bbI_{\cR, \cC} \right)    \left(  \frac{1}{\sqrt{2}}\ket{1}_{\cR'}\otimes\ket{\sigma_{x,\alpha_+}}_{\cS'}\otimes \ket{\psi_{2t}^{\ket{x}\otimes \ket{\alpha}}}_{\cS,\cC} \right)
\right)\leq \sqrt{\gamma}
\end{align}
Now we apply the fact that $\hat{\sF}_{2t}^{\dagger}$ doesn't act on the client's input; the fact that  a unitary operator doesn't change the distance between states; and the fact that tracing out doesn't increase that distance \cite{AKN98}, 
and Eq.~$\eqref{y+ simulator part2}$ 
becomes: 

\begin{align}\label{y+ simulator part3}
\Delta\left(    
\ket{\tilde{\psi}_{2t}^{\ket{x}\otimes \ket{\alpha} }}  ,
\left( \hat{\sF}_{2t}^\dagger \otimes \bbI \right)    \left(\ket{\sigma_{x\alpha_+}}\otimes\ket{\psi_{2t}^{\ket{x}\otimes \ket{\alpha} }} \right)
\right)\leq \sqrt{2\gamma}
\end{align}
Similarly, by projecting onto $\ket{0}_{\cR'}$ instead of $\ket{1}_{\cR'}$ in the derivation of $\ref{y+ simulator part2}$, we get 

\begin{align}\label{y+ simulator part4}
\Delta\left(    
\ket{\tilde{\psi}_{2t}^{\ket{x_0}\otimes \ket{0}}}  ,
\left( \hat{\sF}_{2t}^\dagger \otimes \bbI \right)    \left(\ket{\sigma_{x\alpha_+}}\otimes
\ket{\psi_{2t}^{\ket{x_0}\otimes \ket{0} }} 
\right)
\right)\leq \sqrt{2\gamma}
\end{align}
We now want to apply the triangle inequality to ~\eqref{y+ simulator part4}, using Eqs.~\eqref{0 simulator part1}. Applying yet again the same sequence of simple argument, namely 
the fact that unitary transformations preserve the trace distance and tracing out can only decrease it, we get
\begin{align}\label{y+ simulator part5}
\Delta\left(    
\ket{\sigma_{x_0,0}} ,
\ket{\sigma_{x\alpha_+}}
\right)\leq 2\sqrt{2\gamma}
\end{align}
And we can use Eq.~\eqref{y+ simulator part5} together with Eq.~\eqref{y+ simulator part3} to get:
\begin{align}\label{y+ simulator final}
\Delta\left(    
\ket{\tilde{\psi}_{2t}^{\ket{x}\otimes \ket{\alpha} }}  ,
\left( \hat{\sF}_{2t}^\dagger \otimes \bbI \right)    \left(\ket{\sigma_{x_0,0}}\otimes\ket{\psi_{2t}^{\ket{x}\otimes \ket{\alpha} }} \right)
\right)\leq 3\sqrt{2\gamma}
\end{align}

And finally combine Eq.~\eqref{y+ simulator final}, \eqref{0 simulator part2} and \eqref{adversarial simulator} (in a similar way to how we derived Eq.~\eqref{0 simulator final})  to get:
\begin{align}
\label{y simulator final1}
\Delta\left(
\tr_{\cB_0}\left(       \left( \tilde{\sI}_t   \otimes \mathbb{I}\right)   \left(\ket{x}\bra{x}\otimes\ket{\alpha}\bra{\alpha}\otimes \rho_{joint} \right) \right) ,
\tr_{\cB_t}\left(  \ket{\tilde{\psi}_t^{\ket{x}\otimes \ket{\alpha}}}  \right)
\right) \leq \epsilon + 3\sqrt{2\gamma}.
\end{align}

This finishes our proof. 
\end{proof}

%% file: lowerbound.tex

\section{Linear Lower Bound in the Specious Model, Even with Prior Entanglement}
\label{sec:lb}
\newcommand{\rA}{\mathscr{A}}
\newcommand{\rB}{\mathscr{B}}
\newcommand{\rC}{\mathscr{C}}
\newcommand{\rD}{\mathscr{D}}
\newcommand{\rE}{\mathscr{E}}
\newcommand{\rF}{\mathscr{F}}

In this section we show that in the standard specious model, even allowing arbitrarily long prior entanglement, it is still impossible to achieve QPIR with sublinear communication. We do so by presenting a new lower bound argument based on an interactive leakage chain rule in~\cite{LC18c}, which allows us to establish linear lower bounds on both the server's communication complexity and the total communication complexity in a unified way. Then we observe that the lower bound on the server's communication complexity extends trivially to the case with arbitrary prior entanglement. In the following, we state some useful preliminaries in Section~\ref{sec:LB-background} and  present our lower bound in Section~\ref{subsec:lower-bound}.


\subsection{Quantum Information Theory Background} \label{sec:LB-background}

We first recall the notion of quantum min-entropy. Consider a bipartite quantum state $\rho_{AB}$.
The  quantum  min-entropy of $A$ conditioned on $B$ is defined as
\begin{align*}
H_{\min}(A|B)_{\rho} =  -\inf_{\sigma_{B}}
\left\{
\inf \left\{\lambda\in \mathbb{R}: \rho_{AB}\leq 2^\lambda I_{A}\otimes \sigma_B\right\}
\right\}.
\end{align*}
When $\rho_{AB}$ is a cq-state (i.e., the $A$ register is a classical state), the quantum min-entropy has a nice operational meaning in terms of guessing probability~\cite{KRS09}. Specifically, if $H_{\min}(A|B)_{\rho} = k$, then the optimal probability of predicting the value of $A$ given $\rho_B$  is exactly $2^{-k}$.

In the following, we state the interactive leakage chain rule in~\cite{LC18c}. Let $\rho = \rho_{AB}$ be a cq-state, that is, the system $A$ is classical while $B$ is quantum. The interactive leakage chain rule bounds how much the min-entropy $H_{\min}(A|B)_{\rho}$ can be decreased by an ``interactive leakage'' produced by applying a two-party protocol $\Pi =  \{\rA,\rB,\rho_{joint},s\}$ to $\rho$, where $A$ is treated as a classical input to $\rA$ and $B$ is given to $\rB$ as part of its initial state in $\rho_{joint}$. 

\begin{definition} Let $\rho = \rho_{AB}$ be a cq-state. Let $\Pi =  \{\rA,\rB,\rho_{joint},s\}$ be a two-party protocol where $\rho_{joint}$ contains $\rho_B$ in the $\rB_p$ system, and $\rho_{in}$ be an input state where the classical state $\rho_A$ is copied to $\cA_0$ as the input for $\rA$. (That is, $\cA_0$ has an initial state $\ket{0}_{A_0}$ and we do controlled NOT gates from $\rho_A$ to $\ket{0}_{A_0}$.) Consider the protocol execution $\left[\mathscr{A}\circledast\mathscr{B}\right](\rho_{in})$ and let $\sigma_{AB_s}$ be the final state where $A$ denotes the original classical state and $B_s$ denotes the final state of $\rB$. We say $\sigma_{B_s}$ is an \emph{interactive leakage} of $A$ produced by $\Pi$.
\end{definition}



\begin{theorem} \label{lemma:interactiveLCR2}
Let $\rho = \rho_{AB}$ be a cq-state. Let $\sigma_{AB_s}$ be the final state of a two-party protocol $\Pi =  \{\rA,\rB,\rho_{joint},s\}$ with certain input state $\rho_{in}$. Let $m_A$ and $m_B$ be  the communication complexity of $\rA$ and $\rB$, respectively. We have
	\begin{align}
	&H_{\min}(A|B_s)_{\sigma}\geq H_{\min}(A|B)_\rho -  \min\{ m_A+m_B, 2m_{A}\}, \label{eq:LCR}
	\end{align}
\end{theorem}


We will also use the following lemma about gentle measurement, which is first proved by Winter \cite{Win99}
and improved by Ogawa and Nagaoka~\cite{ON07},  and is also referred to as the almost-as-good-as-new Lemma by Aaronson~\cite{Aar04}.
It says that the post-measurement state of an almost-sure measurement will remain close to its original.
The following version is taken from Wilde's book~\cite{Wilde13}.
\begin{lemma}  \label{lemma:gentle}
 Suppose $0\leq \Lambda\leq I$ is a measurement operator such that for a mixed state $\rho$,
 \[\tr\left(\Lambda \rho\right)\geq 1 - \epsilon.\]
 Then  the post-measurement state $\tilde{\rho}$  is $\sqrt{\epsilon}$-close to the original state $\rho$:$$||\tilde{\rho}-\rho||_{\tr}\leq \sqrt{\epsilon}.$$
 \end{lemma}

We will also need the following lemma, which can be proved by a standard argument using Uhlmann theorem and the Fuchs and van~de Graaf~inequality~\cite{FvdG99} (for a proof, see, e.g.,~\cite{BB15}).


 \begin{lemma}\label{lemma:lo_attack}
 Suppose $\rho_A$, $\sigma_A\in\cA$ are two quantum states with  purifications $\ket{\phi}_{AB}$, $\ket{\psi}_{AB}\in \cA\otimes\cB$, respectively, and $||\rho_A-\sigma_A||_{\tr}\leq \epsilon$. Then there exists a unitary $U_B\in L(\cB)$ such that
 $$||\ket{\phi}_{AB}-I_A\otimes U_B\ket{\psi}_{AB}||_{\tr}\leq \sqrt{\epsilon(2-\epsilon)}.$$
 \end{lemma}

 \subsection{Our Lower Bound} \label{subsec:lower-bound}

\begin{theorem} \label{thm:lower-bound}
Let $\Pi=\{\rA,\rB,\rho_{joint}=\ket{0}\bra{0},s\}$ be a  QPIR protocol for the server's database of size $n$. Suppose $\Pi$ is  $(1-\delta)$-correct and  $\epsilon$-private against $\gamma$-specious servers with $\delta \leq n^{-4}/100, \epsilon \leq n^{-8}/100$.  Then the server's communication complexity is at least $(n - 1)/2$ and the total communication complexity is at least $n - 1$.
\end{theorem}

In the above theorem, we consider protocols with no prior setup, i.e., $\rho_{joint}=\ket{0}\bra{0}$. We observe that the lower bound for the server's communication complexity extends  for general $\rho_{joint}$, since one can think of $\rho_{joint}$ as prepared by the client, who sends the server's initial state to the server at the beginning of the protocol. This simple reduction does not increase the server's communication complexity and extends the lower bound on the server's communication complexity  for arbitrary $\rho_{joint}$.

\begin{corollary}
Let $\Pi=\{\rA,\rB,\rho_{joint},s\}$ be a  QPIR protocol for the server's database of size $n$ with arbitrary $\rho_{joint}$. Suppose $\Pi$ is  $(1-\delta)$-correct and  $\epsilon$-private against $\gamma$-specious servers with $\delta \leq n^{-4}/100, \epsilon \leq n^{-8}/100$.  Then the server's communication complexity is at least $(n - 1)/2$.
\end{corollary}



We now prove Theorem~\ref{thm:lower-bound}.

\begin{proof} To establish communication complexity lower bound for $\Pi$, we consider a purified version $\bar{\Pi}=\{\bar{\rA},\bar{\rB},\rho_{joint},s\}$ of $\Pi$, where both parties' operations are purified. Specifically, $\bar{\rA}$ is modified from $\rA$, where the sequence of quantum operations $\bar{\mathscr{A}}_1,\dots,\bar{\mathscr{A}}_s$ are unitaries
  \begin{align*}
 \bar{\mathscr{A}}_1:& L(\cA_0\otimes \bar{\cA}_0)\rightarrow L(\cA_1\otimes \bar{\cA}_1 \otimes \cX_1),\\
 \bar{\mathscr{A}}_t:& L(\cA_{t-1}\otimes \bar{\cA}_{t-1}\otimes \cY_{t-1})\rightarrow L(\cA_{t}\otimes \bar{\cA}_t\otimes \cX_{t}), t=2,\dots,s;
 \end{align*}
$\bar{\cA}_0$ is  of sufficiently large dimension and  initialized to $\ket{0}$; 
  $\bar{\cA}_t$ are called purifying spaces and
\[ \tr_{\bar{\cA}_t} (\bar{\rho}_t(\rho_{in}) )=\rho_t(\rho_{in}) \]
for all $\rho\in \cA_0\otimes \cB_0$.
The purified  $\bar{\mathscr{B}}$ for $\mathscr{B}$
is similarly defined.

By inspection, it is easy to verify that $\bar{\Pi}$ preserves the properties of $\Pi$, i.e., $\bar{\Pi}$ is also $(1-\delta)$-correct,  $\epsilon$-private against $\gamma$-specious servers, and has the same communication complexity as $\Pi$.  Thus, communication complexity lower bound for $\bar{\Pi}$ implies that for $\Pi$. Also, note that $\bar{\rA}$ is a $0$-specious adversary for $\Pi$.

Now, let us consider an experiment that first samples a uniformly random database $\vc{a} \in \zo^n$, and use $\vc{a}$ as the server's database to run the protocol $\bar{\Pi}$ with an arbitrary fixed input of the client. Note that execution of the protocol  can be viewed as producing an interactive leakage of $\vc{a}$. Let $\rho_{AB}$ denote the final state where system $A$ denotes the input $\vc{a}$ and system $B$ has the client's final local state. By Theorem~\ref{lemma:interactiveLCR2}, we have
$$H(A|B)_\rho \geq H(A)_\rho - \min \{2 m_A, m_A + m_B\},$$
where $m_A, m_B$ denote the server and the client's communication complexities, respectively. The operational meaning of min-entropy says that given the client's state $\rho_B$, one cannot guess the random database $\vc{a}$ correctly with probability higher than $2^{-(H(A)_\rho - \min \{2 m_A, m_A + m_B\})}$. To derive a lower bound on the communication complexity, we show a strategy to predict the database $\vc{a}$ with probability at least $1- n^2 \sqrt{\delta+ 2\sqrt{\epsilon(1-\epsilon)}} > 1/2$, which gives the desired lower bound.



Let $ \sigma_{B}^i=\tr_{A}[\bar{\rA}\circledast\bar{\rB}](\ket{\vc{a}}\bra{\vc{a}}_{A_0}\otimes \ket{i}\bra{i}_{B_0})$
and $ \sigma_{A}^i=\tr_{B}[\bar{\rA}\circledast\bar{\rB}](\ket{\vc{a}}\bra{\vc{a}}_{A_0}\otimes \ket{i}\bra{i}_{B_0})$.

By the definition of privacy, there exists a quantum operation $\rF$ such that
\begin{align}
\Delta\left( \tr_{B_0} \rF_0 \otimes I_{\bar{B}_0} \left(\rho_{in}^1\right),  \sigma_A^1 \right)\leq \epsilon.
\end{align}
Since  $ \tr_{B_0} \rF_0 \otimes I_{\bar{B}_0} \left(\rho_{in}^1\right)= \tr_{B_0} \rF_0 \otimes I_{\bar{B}_0} \left(\rho_{in}^i\right)$ for all $i$,
\begin{align}
\Delta\left(  \tr_{B_0} \rF_0 \otimes I_{\bar{B}_0} \left(\rho_{in}^1\right)-\sigma_A^i\right)\leq \epsilon
\end{align}
We have, by triangle inequality,
\[
\Delta\left( \sigma_A^1-\sigma_A^i\right)\leq 2\epsilon.
\]
for all $i$.

By Lemma~\ref{lemma:lo_attack},  we have
\begin{align}
\Delta\left( I_A\otimes U_B^{1\rightarrow i}\ket{\psi^1}_{\bar{A}\bar{B}}, \ket{\psi^i}_{\bar{A}\bar{B}}\right) \leq 2\sqrt{\epsilon(1-\epsilon)}\triangleq \epsilon', \label{eq:26}
\end{align}
where $\ket{\phi}_{\bar{A}\bar{B}}$ and  $\ket{\psi^i}_{\bar{A}\bar{B}}$ are {purifications} of $\sigma_A^1$ and $\sigma_A^i$, respectively.

By the definition of correctness error, there exists measurement $\cM_i$ such that
\begin{align*}
& \Pr\left\{ \cM_i\left( \sigma_B^i \right)=a_i \right\}\geq 1-\delta.
\end{align*}
Let $$\cM_i'= \left(U_B^{1\rightarrow i} \right)^\dag \cM_i U_B^{1\rightarrow i}$$ for $i=2,\dots,n$. Thus we have by Eq.~(\ref{eq:26})
\begin{align}
& \Pr\left\{ \cM_i'\left( \sigma_B^1 \right)=a_i \right\}\geq 1-\delta-\epsilon'. \label{eq:9}
\end{align}
By Lemma~\ref{lemma:gentle}, the client can recover $\tilde{\sigma}^{(i)}_B$ such that
\begin{align}
\Delta\left(\tilde{\sigma}^{(i)}_B,\sigma_B^1\right)\leq \sqrt{\delta+\epsilon'}. \label{eq:10}
\end{align}

Now we construct a protocol for the client to learn all the bits $\vc{a}=a_1,\dots,a_n$.
First the client chooses input $\ket{1}\bra{1}$. Then he plays the protocol $\bar{\Pi}$ with Alice and obtains $\sigma_B^1$.
Measuring $\sigma_B^1$ by $\cM_1$, the client gets $a_1$ with probability at least $1-\delta$.
By Lemma~\ref{lemma:gentle}, the client can recover $\tilde{\sigma}_B^1$ such that
\[
\Delta\left(\tilde{\sigma}^1_B,\sigma_B^1\right)\leq \sqrt{\delta}.
\]
Then the client measures $\cM_2'$ on $\tilde{\sigma}_B^1$ and then recovers $\tilde{\sigma}_B^2$. Continue this process and
$\tilde{\sigma}_B^k$ will be the state recovered from applying $\cM_k'$ to $\tilde{\sigma}_B^{k-1}$.
We claim that \begin{align}
\Delta\left( \tilde{\sigma}_B^k,\sigma_B^1 \right) \leq k \sqrt{\delta+\epsilon'}. \label{eq:11}
\end{align}
Suppose this is true for $i=2, \cdots, k$.
If we measure $\cM_{k+1}'$ on $\tilde{\sigma}_B^{k+1}$ and on $\sigma_B^1$, respectively,  and recover $\tilde{\sigma}_B^{k+1}$  and $\tilde{\sigma}_B^{(k+1)}$, respectively,
we have  \begin{align}
\Delta\left(\tilde{\sigma}_B^{k+1},\tilde{\sigma}_B^{(k+1)} \right) \leq \Delta\left( \tilde{\sigma}_B^k,\sigma_B^1 \right) \leq k \sqrt{\delta+\epsilon'}      \label{eq:12}
\end{align}
where the first inequality is because quantum operations do not increase trace distance. Now use the triangle inequality with Eqs.~(\ref{eq:10}) and (\ref{eq:12}), and the claim follows by induction.

By Eqs.~(\ref{eq:9}) and  (\ref{eq:11}), the probability of recovering $a_i$ by measuring $\cM_i'$ on $\tilde{\sigma}_B^{i-1}$ is at least $1- i\sqrt{\delta+\epsilon'}$.
Therefore, the client learns $\vc{a}$ with probability at least
\[
\prod_{i=1}^n   \left(1- i\sqrt{\delta+\epsilon'}\right)\geq 1- n^2 \sqrt{\delta +\epsilon'},
\]
which is what we need to complete the proof. 
\end{proof}

%% file: acknowledgments.tex
\paragraph{Acknowledgments}
We thank the anonymous referees for presenting us with the work of Kerenidis et al.~\cite{KLGR16}, and other valuable comments. 
ZB~is supported by the Israel Science Foundation (Grant  468/14), Binational Science Foundation (Grants 2016726, 2014276), and by the European Union Horizon 2020 Research and Innovation Program via ERC Project REACT (Grant 756482) and via Project PROMETHEUS (Grant 780701).
OS~is supported by ERC Grant 280157, by the Israel Science Foundation (Grant 682/18), and by the Cyber Security Research Center at Ben-Gurion University.
CYL~is financially supported from the Young Scholar Fellowship Program by Ministry of Science and Technology
(MOST) in Taiwan, under Grant MOST107-2636-E-009-005. KMC is partially supported by 2016 Academia Sinica Career
Development Award under Grant No. 23-17 and the Ministry of Science and Technology, Taiwan under Grant No. MOST
103-2221- E-001-022-MY3.
DA and AG were supported by ERC Grant 280157 for part of the work on this project, and are supported by the Israel Science Foundation (Grant 1721/17).

%% file: stdprelim.tex

\section{Hilbert Spaces and Quantum States}
\label{sec:stdprelim}

The Hilbert space of a quantum system $A$  is denoted by the corresponding calligraphic letter  $\cA$
and its dimension is denoted by   $\dim(\cA)$.
Let $L(\cA)$ be the space of linear operators on $\cA$.
A  quantum state of system $A$ is described by a \emph{density operator} $\rho_A\in L(\cA)$ that is positive semidefinite  and with unit trace $(\tr(\rho_A)=1)$. 
Let $S(\cA)= \{ \rho_A\in L(\cA): \rho_A\geq 0, \tr(\rho_A)=1\}$ be the set of density operators on  $\cA$.
When $\rho_A\in S(\cA)$ is of rank one, it is called a \emph{pure} quantum state and we can write $\rho=\ket{\psi}\bra{\psi}_A$ for some unit vector $\ket{\psi}_A\in \cA$,
where $\bra{\psi}=\ket{\psi}^{\dag}$ is the conjugate transpose of $\ket{\psi}$. If  $\rho_A$ is not pure, it is called a \emph{mixed} state and can be expressed as a convex combination of pure quantum states.

The Hilbert space of a joint quantum system $AB$ is the tensor product of the corresponding Hilbert spaces $\cA\otimes \cB$.
For $\rho_{AB}\in S(\cA\otimes \cB)$, its reduced density operator in system $A$
is $\rho_A=\tr_B(\rho_{AB})$, 
where
\[
\tr_B(\rho_{AB})= \sum_{i} I_A\otimes\bra{i}_B \left(\rho_{AB} \right)I_A\otimes \ket{i}_B
\]
for an orthonormal basis $\{\ket{i}_B\}$ for $\cB$.
We sometimes use the equivalent notation, 
\[\rho_{AB}|_A:=tr_B(\rho_{AB}).\]

Suppose $\rho_A\in S(\cA)$ of finite dimension $\dim(\cA)$. Then there exists $\cB$ of dimension $\dim(\cB)\geq \dim(\cA)$ and $\ket{\psi}_{AB}\in \cA\otimes \cB$ such that
\[
\tr_B \ket{\psi}\bra{\psi}_{AB} = \rho_A.
\]
The state $\ket{\psi}_{AB}$ is called a purification of $\rho_A$.

The trace distance between two quantum states $\rho$ and $\sigma$ is
$$\Delta(\rho,\sigma)=||{\rho}-{\sigma}||_{\mathrm{tr}},$$
where $||X||_{\mathrm{tr}}=\frac{1}{2}\tr{\sqrt{X^{\dag}X}}$ is the trace norm of $X$.
Hence the trace distance between two pure states $\ket{\alpha}, \ket{\beta}$ is \begin{align}\tdist(\ketbra{\alpha}, \ketbra{\beta}) = \sqrt{1-\abs{\langle \alpha | \beta \rangle}^2}~.\end{align}
\begin{lemma}\label{trace-in lemma}
	Consider a quantum state $\rho_{XY}$ over two registers $X,Y$, and denote $\rho_X = \tr_{Y}(\rho_{XY})$. Then if there exists $\epsilon, \ket{\varphi}$ s.t.\ $\tdist(\rho_X, \ketbra{\varphi}) \le \epsilon$, then there exists $\tilde{\rho}_Y$ s.t.\ $\tdist(\rho_{XY}, \ketbra{\varphi} \otimes \tilde{\rho}_Y) \le \sqrt{\epsilon}$. Furthermore, if $\rho_{XY}$ is pure then so is $\tilde{\rho}_Y$.
\end{lemma}
\begin{proof}
It is sufficient w.l.o.g to prove for a pure $\rho_{XY}$, since it is always possible to purify $\rho_{XY}$ by adding an additional register $Z$, and consider the pure state $\rho_{XYZ}$. The transitivity of the partial trace operation implies that if the theorem is true for $X, (YZ)$, then it is also true for $X,Y$. Also assume w.l.o.g that $\ket{\varphi}=\ket{0}$ (this is just a matter of choosing a basis elements).

Thus we will provide a proof in the case where the joint state of $X,Y$ can be written as a superposition $\ket{\alpha} = \sum_{x,y} w_{x,y} \ket{x}\ket{y}$. Define $P_0 = \Pr[X=0] = \sum_{y} \abs{w_{0,y}}^2$, and note that it must be the case that $P_0 \ge 1-\epsilon$. To see this, note that $P_0$ is the probability of measuring $X=0$ in the experiment where we first trace out $Y$ and then measuring $X$. Since $\tdist(\rho_X, \ketbra{0}) \le \epsilon$, 
the probability of measuring $X=0$ after tracing out $Y$ is $\epsilon$ close to the probability of measuring $X=0$ in $\ketbra{0}$, which is $1$ (see, e.g., \cite{AKN98}). The claim $P_0 \ge 1-\epsilon$ follows.

Now define $\ket{\beta} = \tfrac{1}{\sqrt{P_0}} \sum_y w_{0,y} \ket{y}$, and let $\tilde{\rho}_Y = \ketbra{\beta}$. Then
\begin{align}
\tdist(\rho_{XY}, \ketbra{0} \otimes \tilde{\rho}_Y) = \tdist(\ketbra{\alpha}, \ketbra{0} \otimes\ketbra{\beta}) = \sqrt{1-\abs{\langle \alpha | (0 , \beta) \rangle}^2}~.
\end{align}
We have 
\begin{align}
\langle \alpha | (0 , \beta) \rangle = \tfrac{1}{\sqrt{P_0}} \sum_{y} \abs{w_{0,y}}^2 = \sqrt{P_0}~,
\end{align}
which implies that indeed $\tdist(\rho_{XY}, \ketbra{0} \otimes \tilde{\rho}_Y) = \sqrt{1-P_0} \le \sqrt{\epsilon}$.
\end{proof}

%% file: logscheme.tex

\def\dbvec{{\db}}

\section{Security Analysis of Kereneidis et al.'s Protocol}
\label{sec:protocol}

For completeness, we restate\footnote{We make one minor adaptation -- see Remark~\ref{remark:reusableEbit}.} the QPIR protocol with pre-shared entanglement by Kerenidis et al.~\cite[Section 6]{KLGR16}. 
Given a database $\dbvec \in\binset^{\dblen}$ for some ${\dblen}=2^{\dbllen}$ as input to the server, and index $i \in [{\dblen}]$ as input to the client (If the client's input is a superposition, the algorithm is run in superposition), we denote the protocol $\Pi_{\dblen}$ as follows. 

The protocol $\Pi_{\dblen}$ is recursive and calls $\Pi_{{\dblen}/2}$ as a subroutine. For the execution of $\Pi_{\dblen}$, the parties are required to pre-share a pair of entangled state registers $\tfrac{1}{2^{{\dblen}/4}}\sum_{\auxr \in \binset^{{\dblen}/2}} \ket{\auxr}_R \ket{\auxr}_{R'}$, where $R$ is held by the server and $R'$ is held by the client. They also share an entangled state needed for the recursive application of the protocol $\Pi_{{\dblen}/2}$ (and the recursive calls it entails). Unfolding the recursion, this means that for all ${\dblen}'=2^{{\dbllen}'}$ with ${\dbllen}' \in [{\dbllen}-1]$, there is an entangled register of length ${\dblen}'$ shared between the client and the server.

The protocol execution is described in shorthand Figure~\ref{fig:logscheme}. In what follows we provide a detailed description and analyze the steps of the protocol to establish correctness and assert properties that will allow us to analyze privacy.

\begin{enumerate}
	\item 	If ${\dblen}=1$ then the database contains a single value. In this case there is no need for shared entanglement, and the server sends a register $F$ containing $\ket{\dbvec}$ (the final response) to the client, and the protocol terminates. This is trivially secure and efficient. Otherwise proceed to the next steps.
	
	\item The server denotes $\dbvec_0, \dbvec_1 \in \binset^{{\dblen}/2}$ s.t.\ $\dbvec= [\dbvec_0 \| \dbvec_1]$, i.e. the low-order and high-order bits of the database respectively. The server starts with two single-bit registers $Q_0, Q_1$ initialized to $0$. The server CNOTs $Q_b$ with the inner product of $R$ and $\dbvec_b$ so that it contains $\ket{\auxr \cdot \dbvec_b}_{Q_b}$, and sends $Q_0, Q_1$ to the client.
	
	At this point, the joint state between the client and (an honest) server is \[\sum_{\auxr \in \binset^{{\dblen}/2}} \ket{\auxr}_R \ket{\auxr}_{R'} \ket{\auxr \cdot \dbvec_0}_{Q_0} \ket{\auxr \cdot \dbvec_1}_{Q_1}~.\]
	
	In particular the reduced density matrix of the server's state is independent of the index ${i}$.
	
	\item Let $b^*=\lfloor \tfrac{i-1}{{\dblen}} \rceil$ denote the most significant bit of $i$. The client evaluates a $Z$ gate on $Q_{b^*}$. It sends $Q_0, Q_1$ back to the server.
	
	At this point, the joint state between the client and (an honest) server is \[\sum_{\auxr \in \binset^{{\dblen}/2}} (-1)^{\auxr \cdot \dbvec_{b^*}} \ket{\auxr}_R \ket{\auxr}_{R'} \ket{\auxr \cdot \dbvec_0}_{Q_0} \ket{\auxr \cdot \dbvec_1}_{Q_1}~.\]
	
	Importantly, the reduced density matrix of the server, which contains the registers $R, Q_0, Q_1$, is the diagonal matrix that corresponds to the classical distribution of sampling a random $\auxr$ in register $R$, and placing $\auxr \cdot \dbvec_0, \auxr \cdot \dbvec_1$ in $Q_0, Q_1$. This density matrix is independent of $b^*$ and therefore of $i$.
	
	\item The server again CNOTs $Q_b$ with the inner product of $R$ and $\dbvec_b$.
	
	At this point, the joint state between the client and (an honest) server is \[\sum_{\auxr \in \binset^{{\dblen}/2}} (-1)^{\auxr \cdot \dbvec_{b^*}} \ket{\auxr}_R \ket{\auxr}_{R'} \ket{0}_{Q_0} \ket{0}_{Q_1}~.\]
	
	From this point on we disregard $Q_0, Q_1$ since they remain zero throughout. Since this step only involves a local unitary by the server, we are guaranteed that its reduced density matrix is still independent of ${i}$.
	
	\item The server performs QFT on $R$ and the client performs QFT on $R'$. The resulting state is
	\begin{align*}
	\tfrac{1}{2^{3 {\dblen}/4}}  \sum_{\auxr,\vc{y},\vc{w} \in \binset^{{\dblen}/2} } (-1)^{\auxr \cdot (\dbvec_{b^*}\oplus\vc{y}\oplus\vc{w})} \ket{\vc{y}}_R \ket{\vc{w}}_{R'} = \tfrac{1}{2^{{\dblen}/4}} \sum_{\vc{y} \in \binset^{{\dblen}/2} } \ket{\vc{y}}_R \ket{\underbrace{\vc{y} \oplus \dbvec_{b^*}}_{\vc{w}}}_{R'}~.
	\end{align*}
	
	Since we only performed local operations on the server and client side (without communication), the server's density matrix remains perfectly independent of $b^*$ and thus of ${i}$.

	\item Note that at this point, the joint state of the client and server is a ``shifted'' entangled state where the shift corresponds to the half-database $\dbvec_{b^*}$ that contains the element that the client wishes to retrieve. More explicitly, $\dbvec[i] = \dbvec_{b^*}[i^*]$ for $i^*=i \pmod{{\dblen}/2}$ contains the $({\dbllen}-1)$ least significant bits of $i$. Therefore, for all $\vc{y}, \vc{w}$ in the support of the joint state, it holds that $\dbvec[i] = \vc{w}[i^*] \oplus \vc{y}[i^*]$.
	
	The client will now ignore (temporarily) the register $R'$ and execute $\Pi_{{\dblen}/2}$ recursively on index ${i}^*$. The (honest) server will carry out the protocol with the value $\vc{y}$ from the register $R$ serving as the server's database. Note that since the register $R'$ is not touched, for the purposes of executing the protocol the value $\vc{w}$ in $R'$ is equivalent to have been measured, and the value $\vc{y}$ in $R$ is equivalent to the deterministic register $\vc{w}\oplus \dbvec_{b^*}$.
	
	We are recursively guaranteed that in the end of the execution of $\Pi_{{\dblen}/2}$, the client receives a register $F$ containing the value $\vc{y}[i^*] = \vc{w}[i^*]\oplus \dbvec_{b^*}[i^*] = \vc{w}[i^*]\oplus \dbvec[i]$. Since the client still maintains the original register $R'$ containing $\vc{w}$, it can CNOT the value $\vc{w}[i^*]$ from $F$ and obtain $\ket{\dbvec[i]}_A$. Namely, in the end of the execution, the register $F$ indeed contains the desired value $\dbvec[i]$.

	\item Lastly, if the client and server desire to ``clean up'' and restore the shared entanglement so that it can be reused in consequent executions, the client can copy the contents of the register $F$ to a fresh register (which is possible since this register contains a classical value). Since the client and server are pure (i.e.\ do not measure) throughout the protocol, they can rewind the execution of the protocol to restore their initial joint entanglement.
	
\end{enumerate}
If the final cleanup step is not executed then the total number of rounds of $\Pi_{{\dblen}}$ is $2 {\dbllen} + 1$, and the total communication complexity is $4l+1$ (recall that ${\dbllen} = \log({\dblen})$). If the cleanup step is executed, the round complexity and communication complexity are doubled due to rewinding the execution.
\begin{remark}\label{remark:reusableEbit}
Note that in the original protocol by Kerenidis et al. step~\ref{step:cleanup} does not appear, and it is not mentioned that the shared entanglement can be cleaned and reused. 
\end{remark}

We conclude that for classical inputs for both the client and the server, the honest server's density matrix is independent of $i$. If the client first measures its input state, privacy holds. But since the very first operation is a CNOT operation (to determine the value of $i$) and since CNOT and measurements in the standard basis commute, we conclude that the server's reduced density matrix is independent of the client's input, even for inputs which are in superposition. 

\begin{lemma}
	The protocol $\Pi_{\dblen}$ is a PIR protocol with perfect correctness and perfect anchored privacy against honest servers.
	 It furthermore has communication complexity $O(\log {\dblen})$, and uses $O({\dblen})$ bits of (reusable) shared entanglement. 
\end{lemma}
\begin{proof}
The analysis in the body of the protocol establishes that the local view of the adversary is independent of $i$, when $i$ is treated as a fixed parameter. For the sake of our privacy notion, we are required to establish that the server's local state is independent of $i$ even when $i$ is an arbitrary quantum state.
This follows since the client refers to the index $i$ as constant, and therefore a superposition over $i$ will translate to a superposition over classical executions of the protocol, each with a fixed $i$. Since the server's local view is identical for any fixed $i$, it will also be in the same state for a superposition, and also for an arbitrary mixed state of $i$ and some potential environment.

The communication complexity and the amount of reusable shared entanglement needed in this protocol follow directly from the protocol. 
\end{proof}

We can therefore apply Theorem~\ref{thm:main} and conclude that $\Pi$ is secure against anchored-specious adversaries.
\begin{corollary}
	There exists a PIR protocol $\Pi$ with logarithmic communication complexity assuming linear shared entanglement, which is perfectly correct and anchored $O(\sqrt{\gamma})$-private against $\gamma$-specious adversaries.
\end{corollary}

\begin{figure}[t]
	\begin{small}
		\begin{center}
			\textbf{Recursive QPIR with Logarithmic Communication}
		\end{center}
		
		\noindent \textbf{Server input:} Database $\dbvec\in\binset^{\dblen}$.\\
		\textbf{Client input:} Index $i \in [{\dblen}]$.\\
				\textbf{Desired output:} Value $\dbvec[i]$ stored in register $F$ on the client side.\\
		\textbf{Setup:} Register $R$ for server and $R'$ for client in joint state $\tfrac{1}{2^{{\dblen}/4}}\sum_{\auxr \in \binset^{{\dblen}/2}} \ket{\auxr}_R \ket{\auxr}_{R'}$.\\ (This setup is for external recursion loop, internal loops require their own $R,R'$ defined recursively.)
		
		\begin{enumerate}
			\item 	If ${\dblen}=1$, copy the (single-bit) database into a register and send to client, then terminate (or go to clean up step \ref{step:cleanup} below).
			
			\item The server denotes $\dbvec_0, \dbvec_1 \in \binset^{{\dblen}/2}$ s.t.\ $\dbvec= [\dbvec_0 \| \dbvec_1]$, i.e. the low-order and high-order bits of the database respectively. The server starts with two single-bit registers $Q_0, Q_1$ initialized to $0$. The server CNOTs $Q_b$ with the inner product of $R$ and $\dbvec_b$ so that it contains $\ket{\auxr \cdot \dbvec_b}_{Q_b}$. It sends $Q_0, Q_1$ to the client. 
			
			\item Let $b^*=\lfloor \tfrac{i-1}{{\dblen}} \rceil$ denote the most significant bit of $i$. The client evaluates a $Z$ gate on $Q_{b^*}$. It sends $Q_0, Q_1$ back to the server. 
			
			\item The server again CNOTs $Q_b$ with the inner product of $R$ and $\dbvec_b$. 
			
			\item The server performs QFT on $R$ and the client performs QFT on $R'$. 
			
			\item Call $\Pi_{{\dblen}/2}$ recursively (with fresh $R,R'$ obtained from the setup). The server input is the contents of the register $R$ (of length ${\dblen}/2$). The client input is $i^*=i \pmod{{\dblen}/2} \in [{\dblen}/2]$. The client receives a response register $F$ as the output of the recursive call. It then CNOTs $R'[i^*]$ into $F$. Finally, $F$ contains the output of the recursive execution.
			
			\item \label{step:cleanup} If it is desired to restore the shared entanglement, copy the (classical) output into a fresh register and rewind the execution of the protocol.
		\end{enumerate}
		
	\end{small}
	\caption{The QPIR protocol of Kerenidis et al.}
	\label{fig:logscheme}
\end{figure}